\documentclass[letterpaper, 10 pt, conference]{ieeeconf}  

\IEEEoverridecommandlockouts                              

\usepackage{graphics} 
\usepackage{epsfig} 

\usepackage{amsmath} 
\usepackage{amssymb}  

\usepackage{amsthm}
\usepackage{tikz}
\usepackage{flushend}
\usepackage[noadjust]{cite}
\newtheorem{thm}{Theorem}
\newtheorem{defn}{Definition}

\newtheorem{lem}{Lemma}

\title{\LARGE \bf
Aggregate Flexibility of Thermostatically Controlled Loads\\ using Generalized Polymatroids
}

\author{Karan Mukhi and Alessandro Abate
\thanks{This work was supported by Réseau de Transport d'Électricité.}
\thanks{Karan Mukhi and Alessandro Abate are with the Department of Computer Science, University of Oxford, UK, OX1 3QG
        {\tt\small \{karan.mukhi,alessandro.abate\}@cs.ox.ac.uk}}%
}

\begin{document}

\maketitle
\thispagestyle{empty}
\pagestyle{empty}

\begin{abstract}
Leveraging populations of thermostatically controlled loads could provide vast storage capacity to the grid. To realize this potential, their flexibility must be accurately aggregated and represented to the system operator as a single, controllable virtual device. Mathematically this is computed by calculating the Minkowski sum of the individual flexibility of each of the devices. Previous work showed how to exactly characterize the flexibility of lossless storage devices as generalized polymatroids-a family of polytope that enable an efficient computation of the Minkowski sum. In this paper we build on these results to encompass devices with dissipative storage dynamics. In doing so we are able to provide tractable methods of accurately characterizing the flexibility in populations consisting of a variety of heterogeneous devices. Numerical results demonstrate that the proposed characterizations are tight.
\end{abstract}

\section{Introduction}

Achieving sustainable power system operation requires deep integration of renewable energy sources. However, many of these sources are inherently intermittent and uncertain in their power output. To economically mitigate these challenges, the grid must be operated with increased flexibility to accommodate the growing share of renewable generation.
In parallel to this, the electrification of heating—combined with the growing adoption of HVAC systems—is expected to significantly increase electricity demand from thermostatically controlled loads (TCLs). As the temperature of these devices is usually allowed to lie within a dead-band around its set-point, TCLs possess a degree of flexibility in their power consumption \cite{Hao2015AggregateLoads}. This flexibility can be harnessed to mitigate intermittency and enable deeper penetrations of renewable generation \cite{Callaway2009TappingEnergy}.

Several methods have been proposed to enable the use of this flexibility in power systems. These approaches broadly fall into three main paradigms: centralized, decentralized, and hierarchical. Centralized approaches involve individual devices communicating their operational constraints to a central coordinator, which then solves a single, system-wide optimization problem. While such methods can guarantee optimality, they generally suffer from poor scalability as the number of devices increases \cite{Tavakkoli2021Bonus-BasedSystem}. Decentralized approaches such as \cite{Tindemans2015DecentralizedResponse, EsmaeilZadehSoudjani2013AggregationAbstractions}, mitigate these scalability challenges by allowing devices to make decisions locally. However, the absence of a centralized decision-making entity makes it challenging to assign accountability for ensuring strict grid reliability guarantees.
Hierarchical approaches strike a balance between the two,  whereby an \textit{aggregator} manages a population of devices and communicates their aggregate flexibility to the system operator \cite{Xu2016HierarchicalLoads}. It is within this paradigm that this our work sits. 

To make use of the flexibility in the populations they control, it is useful for the aggregator to characterize the set of admissible consumption profiles that the population as a whole can take.
The set of consumption profiles that an individual TCL can take whilst satisfying its operational constraints can be represented as a
convex polytope \cite{Hao2015AggregateLoads}. The aggregate flexibility of a population of devices is geometrically represented by the Minkowski sum of their individual flexibility polytopes \cite{Barot2017APolytopes}. While both active and reactive power flexibility are considered in works such as \cite{Kundu2018ApproximatingApproach} and \cite{Zhou2024AggregatedGuarantee}, the majority of the literature focuses on the aggregation of active power flexibility. In general, calculating the Minkowski sum is computationally expensive and so many of the works in the literature have been dedicated to inner approximations of the sum. Many of these methods introduce a base polytope to approximate individual flexibility sets. These base sets are selected to ensure that computing the Minkowski sum over a collection of the approximations is computationally efficient. In \cite{Zhao2017ALoads} \textit{homothets} are used, and approximations are made by scaling and translating the homothet, with \cite{Taha2024AnPopulations} extending this to affine transformations of the homothet. Similarly, in \cite{Muller2019AggregationResources} \textit{zonotopes} are used as the base set. \textit{Permutahedra}  are used in \cite{Mukhi2023AnVehicles}  and \cite{Panda2024EfficientVehicles} when considering a populations of EVs, with an extension to stochastic settings in \cite{Mukhi2024DistributionallyFlexibility}. 

Recent work established exact characterizations of individual flexibility sets as \textit{generalized polymatroids} (g-polymatroids). These are a rich class of polytopes that possess several desirable structural properties. Of particular relevance in this context is the fact that their Minkowski sum can be computed efficiently. Leveraging this property, the authors of \cite{Mukhi2025ExactPolymatroids} showed how to obtain exact and tractable representations of the aggregate flexibility as the Minkowski sum of individual sets, enabling scalable, exact characterization across large populations of devices. However, the assumptions adopted in that work constrain the characterization to devices with lossless storage dynamics. While such assumptions are well-suited for modeling the flexibility of storage-based resources such as electric vehicles (EVs) and energy storage systems (ESSs), they are not appropriate for representing the flexibility inherent in TCLs, where energy dissipation and thermal losses play a significant role. In this paper, we extend these methods by relaxing the assumption of lossless storage dynamics, thereby generalizing the framework to accommodate flexibility characterizations for TCLs. This generalization enables a unified approach for accurately characterizing the flexibility of heterogeneous device populations within a common aggregation framework.

To achieve this, we propose a base polytope that conforms to the structure of a g-polymatroid, which we use to construct maximal inner approximations of the flexibility sets of a TCL. The base polytopes we employ are more expressive than those utilized in existing literature, resulting in tighter approximations of individual flexibility sets and, consequently, a more accurate characterization of aggregate flexibility. Relying on the results in \cite{Mukhi2025ExactPolymatroids} these polytopes can be aggregated with the flexibility sets of other TCLs \textit{and} other flexible devices. In summary, the main contributions of this work can be summarized as follows: first, we show how to derive maximal inner approximations of TCL flexibility sets using g-polymatroids. Secondly, we show how this characterization fits within existing g-polymatroid-based aggregation methods.

The rest of this paper is structured as follows: in Section~\ref{sec:prob_form}, we formulate the problem by modeling TCL dynamics, defining flexibility sets, and introducing g-polymatroids as a tool to tractably compute the aggregate flexibility sets. Section~\ref{sec:approximation} presents the construction of maximal inner g-polymatroid approximations for the individual flexibility sets and shows how these approximations can be aggregated efficiently. In Section~\ref{sec:num_res} we benchmark our methods against other approximations in the literature, and provide a case study demonstrating the utility of this work, finally we draw conclusions and provide directions for future work in Section~\ref{sec:conc}.
\vspace{5pt}
\subsubsection*{Notation}
For a vector $u \in \mathbb{R}^\mathcal{T}$ where $\mathcal{T}$ is a finite set and $t \in \mathcal{T}$, we let $u(t)$ denote the $t^{th}$ component of $u$.
Given a set $A \subseteq \mathcal{T}$, we let $x(A) = \sum_{t \in A} u(t)$. 
Lastly, $\sum (\cdot)$ denotes both sums and Minkowski sums, depending on the context.

\section{Problem Formulation}\label{sec:prob_form}
In this section, we present a model for TCL power consumption and define our notions of individual and aggregate flexibility sets. We provide a brief introduction to g-polymatroids, and formalize the problem of finding maximal inner the aggregate flexibility set.
In the following, we consider an aggregator that exerts direct control over the power consumption of a population of TCLs, indexed by $i \in \mathcal{N} := \{1, \dots, N\}$. The problem is formulated over a discrete time horizon comprising $T$ periods, indexed by $t \in \mathcal{T} := \{1, \dots, T\}$, where (for simplicity and without loss of generality) each period has a fixed duration of unit length.

\subsection{TCL Consumption Model}
The temperature dynamics of a TCL can be modeling using the following linear model \cite{Zhao2017ALoads}:
\begin{equation}
    \theta(t) = a\theta(t-1) + (1-a)(\theta_a(t) - b p(t)),
\end{equation}
where $\theta(t)$ is the temperature and $p(t)$ the power consumption of the TCL in timestep $t$. $\theta_a(t)$ is the ambient temperature which may be time-dependent, however for simplicity we assume it is constant in the following
and $a \in [0,1)$ and $b$ are model parameters that depend on the thermal properties of the device. We have assumed the device cools, however a similar model can be applied to heating devices. The temperature of the device must stay within a dead-band, $\Delta$, around the set-point temperature $\theta_r$ of the device:
\begin{equation}
    \theta_r - \Delta/2 \leq \theta(t) \leq \theta_r + \Delta/2 \quad \forall t.
\end{equation}
In this model we assume the power consumption $p(t)$ is a continuous variable bounded by the power rating of the device, $\overline{p}$, i.e. $p(t) \in [0,\overline{p}]$. Typically, TCLs tend to be have on/off controllers, and so $p(t)$ is better modeled as a binary variable: $p(t) \in \{0, \overline{p}\}$, however computing the flexibility of these non-linear models is tricky. When considering large populations of devices, modeling power consumption as continuous leads to a good approximation \cite{Hao2015AggregateLoads}.

After a change of variables, the temperature dynamics can be reformulated as \cite{Zhao2017ALoads}:
\begin{equation}
    x(t) = a x(t-1) + u(t).
\end{equation}
The system is subject to the follow constraints on its state:
\begin{equation}
    \underline{x} \leq x(t) \leq \overline{x} \quad \forall t,
\end{equation}
that ensure the temperature limits of the device are not violated. We can re-write this in terms of $u(t)$ 
\begin{equation*}
    \underline{x}(t) \leq \sum_{s=1}^t a^{t-s}u(s) \leq \overline{x}(t) \quad \forall t,
\end{equation*}
where $\underline{x}(t) = \underline{x} - a^tx(0)$ and $\overline{x}(t) = \overline{x} - a^tx(0)$.
The control input, $u(t)$, is constrained to lie within the interval
\begin{equation}
    \underline{u} \leq u(t) \leq \overline{u} \quad \forall t.
\end{equation}

\subsection{Flexibility Polytopes}
Now given a device with parameters $(a_i, \underline{x}_i, \overline{x}_i, \underline{u}_i, \overline{u}_i)$, we would like to characterize the set consumption profiles it may take without violating its temperature or power constraints.
\begin{defn}
    For a TCL with consumption parameters $(a_i, \underline{x}_i, \overline{x}_i, \underline{u}_i, \overline{u}_i)$, the individual flexibility set of the device, denoted $\mathcal{F}_i$ is the set of all feasible consumption profiles for the TCL:
    \begin{equation*}
        \mathcal{F}_{i} := \left\{ u \in \mathbb{R}^{\mathcal{T}} \middle\vert\;
        \begin{aligned}
        \underline{u}_i(t) &\leq u_i(t) \leq \overline{u}_i(t), \quad \forall t \in \mathcal{T} \\
        \underline{x}_i(t) & \leq \sum_{s=1}^t a_i^{t-s}u_i(s) \leq \overline{x}_i(t), \forall t \in \mathcal{T}
        \end{aligned}
        \right\}
    \end{equation*}
\end{defn}
As defined above, $\mathcal{F}_i$ is characterized by a set of linear constraints. Due to the limits on $u(t)$, it is clearly bounded and so the individual flexibility sets for TCLs are all compact convex polytopes. 

Now, we consider an aggregator that has direct control over the consumption of a collection of $N$ devices, each with their own consumption parameters. To make use of the flexibility in the population they control, the aggregator needs to characterize the set of \textit{aggregate} consumption profiles the populations as a whole may take \cite{Taha2024AnPopulations}.
\begin{defn}
    The aggregate flexibility set of a population of devices, is the set of all feasible aggregate consumption profiles of the population:
    \begin{equation*}
        \mathcal{F}_\mathcal{N}:= \left\{ u_\mathcal{N} \in \mathbb{R}^{\mathcal{T}} \middle| u_\mathcal{N} = \sum_{i \in \mathcal{N}} u_i,\;\; u_i \in \mathcal{F}_i \; \forall i \in \mathcal{N} \right\}.
    \end{equation*}
\end{defn}
$\mathcal{F}_\mathcal{N}$ is, by definition, the Minkowski sum of the individual flexibility sets: 
\begin{equation}\label{eq:m_sum}
    \mathcal{F}_\mathcal{N} = \sum_i^N \mathcal{F}_i.
\end{equation}
In general computing the Minkowski sum of a collection of polytopes is NP-hard \cite{Tiwary2008OnPolytopes}, and so an exact computation of \eqref{eq:m_sum} is intractable. In this work, similar to \cite{Zhao2017ALoads}, we instead aim at providing maximal inner approximations of $\mathcal{F}_\mathcal{N}$.

\subsection{Generalized Polymatroids}
To provide a method of approximating the aggregate flexibility set,  we will make use of g-polymatroids. These are an expressive family of polytopes that exhibit a number of useful properties. Notably for this work, computing the Minkowski sum of a collection of them is efficient. Below we provide a brief definition of their construction and introduce a relevant theorem that will aid us in the computing \eqref{eq:m_sum}, the reader is referred to \cite{Frank2011ConnectionsOptimization} for a complete treatment of these polytopes.

A \textit{submodular function}, $b: 2^\mathcal{T} \rightarrow \mathbb{R}$, is a set function defined over subsets of a finite set $\mathcal{T}$, that satisfies the submodular property:
\begin{equation*}
    b(A) + b(B) \geq b(A\cup B) + b(A \cap B).
\end{equation*}
Similarly, a \emph{supermodular function}, $p$, can be defined by reversing the inequality in the equation above, or equivalently, as the negative of a submodular function; that is, $p = -b$ is supermodular if $b$ is submodular. A pair of super- and submodular functions can be used to generate a g-polymatroid.

\begin{defn}\cite{Frank2011ConnectionsOptimization}
    Given a pair of super- and submodular functions $(p,b)$, a g-polymatroid, denoted $\mathcal{Q}(p,b)$, is defined as
    \begin{equation*}
        Q(p,b) := \left\{ u \in  \mathbb{R}^{\mathcal{T}} \mid p(A) \leq u(A) \leq b(A) \;\; \forall A \subseteq \mathcal{T} \right\}.
    \end{equation*}
\end{defn}
G-polymatroids are an expressive class of polytope, that can exactly represent a wide range of polytopes. In \cite{Mukhi2025ExactPolymatroids} it was shown how they can exactly characterize the flexibility sets for storage devices with lossless dynamics.
\begin{thm}\label{thm:g-poly_m_sum}The Minkowski sum of a collection of g-polymatroids is given by
    \begin{equation}
        \mathcal{Q}\left( \sum_i^N p_i, \sum_i^N b_i\right) = \sum_i^N  \mathcal{Q}(p_i, b_i)
    \end{equation}
\end{thm}
This result provides us with a method of computing the exact Minkowski sum for a collection of g-polymatroids.
By approximating the flexibility sets of individual TCLs as g-polymatroids, we can use Theorem \ref{thm:g-poly_m_sum} to tractably compute approximations for the aggregate flexibility of a collection of TCLs. Specifically, we want to find maximal inner g-polymatroid approximations of the individual flexibility sets, i.e.
\begin{equation}
    \mathcal{Q}(p_i, b_i) \subset \mathcal{F}_i
\end{equation}
Here $p_i$ and $b_i$ are the super- and submodular functions generating the maximal inner g-polymatroid of $\mathcal{F}_i$. Using these approximations and defining 
\begin{equation*}
        p_\mathcal{N} := \sum_i^N p_i, \quad  \quad  b_\mathcal{N} := \sum_i^N b_i,
\end{equation*}
(and similarly for $b^n_\mathcal{N}$ and $b^n_\mathcal{N}$), we can employ Theorem \ref{thm:g-poly_m_sum} to provide inner and outer approximations of the aggregate flexibility sets:
\begin{equation}
    \mathcal{Q}\left( p_\mathcal{N}, b_\mathcal{N}\right) \subseteq \mathcal{F} 
\end{equation}
Hence, the rest of this paper is focused on finding the super- and submodular functions that generate the maximal inner, and minimal outer g-polymatroid approximations of $\mathcal{F}_i$.
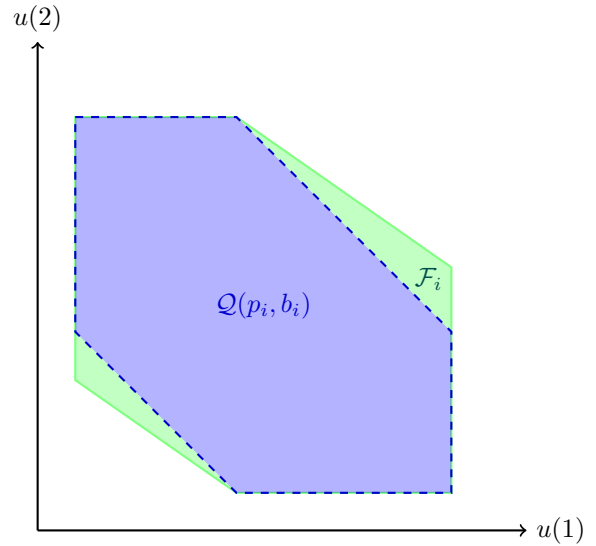
\begin{figure}
    \centering
        \begin{tikzpicture}[xscale=5, yscale=5]
    

        \pgfmathsetmacro{\xmax}{1.3}
        \pgfmathsetmacro{\xmin}{0.3}
        \pgfmathsetmacro{\a}{0.7}
    
        \pgfmathsetmacro{\umax}{1}
        \pgfmathsetmacro{\umin}{0}
        
        \pgfmathsetmacro{\uone}{(\xmax - \umax)/\a}
        \pgfmathsetmacro{\utwo}{(\xmax - \a*\umax)}
    
        \pgfmathsetmacro{\uonel}{(\xmin - \umin)/\a}
        \pgfmathsetmacro{\utwol}{(\xmin - \a*\umin)}

        \coordinate (v1) at (\umin,\umax);
        \coordinate (v2) at (\uone,\umax);
        \coordinate (v3) at (\umax,\utwo);
        \coordinate (v4) at (\umax,\umin);
        \coordinate (v5) at (\uonel, \umin);
        \coordinate (v6) at (\umin, \utwol);
    
        \coordinate (a1) at (\umin,\umax);
        \coordinate (a2) at (\uone,\umax);
        \coordinate (a3) at (\umax,\uone);
        \coordinate (a4) at (\umax,\umin);
        \coordinate (a5) at (\uonel, \umin);
        \coordinate (a6) at (\umin, \uonel);
    
        \coordinate (b1) at (\umin,\umax);
        \coordinate (b2) at (\utwo,\umax);
        \coordinate (b3) at (\umax,\utwo);
        \coordinate (b4) at (\umax,\umin);
        \coordinate (b5) at (\utwol, \umin);
        \coordinate (b6) at (\umin, \utwol);


    \fill[green!80, opacity=0.3] (v2) -- (v3) -- (a3) -- cycle;
    \fill[green!80, opacity=0.3] (v5) -- (v6) -- (a6) -- cycle;
    
    \draw[green!50, thick] (v1) -- (v2) -- (v3) -- (v4) -- (v5) -- (v6) -- cycle;
    
    \fill[blue, opacity=0.3] (a1) -- (a2) -- (a3) -- (a4) -- (a5) -- (a6) -- cycle;
    \draw[blue!80!black, dashed, thick] (a1) -- (a2) -- (a3) -- (a4) -- (a5) -- (a6) -- cycle;
    %
    
    \draw[->, thick] (-0.1, -0.1) -- (1.2, -0.1) node[anchor=west] {$u(1)$};
    \draw[->, thick] (-0.1, -0.1) -- (-0.1, 1.2) node[anchor=south] {$u(2)$};
  \node[blue!80!black] at (0.5, 0.5) {$\mathcal{Q}(p_i, b_i)$};

  \node[teal!60!black] at (0.94, 0.57) {$\mathcal{F}_i $};


    \end{tikzpicture}
    \caption{The individual flexibility set of a TCL $\mathcal{F}_i$ (blue shaded region) and its maximal inner approximation as a g-polymatroid $\mathcal{Q}(p_i, b_i)$ (green shaded region).}
    \label{fig:enter-label}
\end{figure}

\section{G-Polymatroid Approximations}\label{sec:approximation}
In this section, we derive maximal inner approximations of the individual flexibility sets using 
g-polymatroids. To this end, we first introduce a base polytope that satisfies the structural properties of a g-polymatroid. We then formulate an optimization problem to determine the optimal parameters of this base polytope, such that it yields the largest possible inner approximation of the true flexibility set. From this we use the results in \cite{Mukhi2025ExactPolymatroids} to derive the associated supermodular and submodular functions that define the g-polymatroid. 
Leveraging their efficient Minkowski sum computation, we construct inner approximations of the aggregate flexibility set.
\subsection{Base Polytope}
We start by introducing a base set that will be used to approximate $\mathcal{F}_i$. A common approach in the literature is to use a base set defined as a polytope representing a generalized battery model. In this work we also have the added constraint that the base set must also be a g-polymatroid, so that we can make use of Theorem, \ref{thm:g-poly_m_sum}.  In light of this we use a base set of the following form:
    \begin{equation*}
        \mathcal{B}_i(\underline{y}, \overline{y}) := \left\{ u \in \mathbb{R}^{\mathcal{T}} \middle\vert\;
        \begin{aligned}
        \underline{u}_i &\leq u(t) \leq \overline{u}_i,  \forall t \in \mathcal{T} \\
        \underline{y}(t) &\leq \sum_{s=1}^t u(s) \leq \overline{y}(t), \forall t \in \mathcal{T}
        \end{aligned}
        \right\}
\end{equation*}
parametrized by $\underline{y}, \overline{y} \in \mathbb{R}^\mathcal{T}$.
This set is sufficiently expressive to capture much of the flexibility of the individual sets, including power and time-dependent energy constraints, and aligns closely with the generalized battery models commonly in the literature. Moreover, as shown in \cite{Mukhi2025ExactPolymatroids}, polytopes of the form of $\mathcal{B}_i$ are g-polymatroids.
However, we note that alternative g-polymatroids may more effectively capture the dissipative dynamics of leaky storage devices, particularly over extended time horizons, finding such base sets is left as a direction for future work.
Given the base polytope $\mathcal{B}_i$, our objective is to compute a maximal inner approximation of the feasible set $\mathcal{F}_{i}$ using $\mathcal{B}_i(\underline{y}, \overline{y})$. This entails identifying the optimal time-dependent energy bounds $\underline{y}, \overline{y} \in \mathbb{R}^{\mathcal{T}}$ such that $\mathcal{B}_i$ is contained within $\mathcal{F}_i$. Unlike the homothet-based methods commonly employed in the literature, which rely on scaling and translating a fixed prototype set, our approach directly optimizes over the parameters $\underline{y}$ and $\overline{y}$. This parameterization provides additional degrees of freedom, enabling the derivation of tighter approximations of $\mathcal{F}_{i}$.

\begin{figure}
    \centering
    \begin{tikzpicture}[xscale=4, yscale=4]

    \pgfmathsetmacro{\umax}{1}
    \pgfmathsetmacro{\xmax}{1.1}
    \pgfmathsetmacro{\a}{0.5}
    \pgfmathsetmacro{\d}{0.2}
    \pgfmathsetmacro{\uone}{(\xmax - \umax)/\a}
    \pgfmathsetmacro{\utwo}{(\xmax - \a*\umax)}

    \pgfmathsetmacro{\xonex}{(\xmax -(\umax + \d))/\a}
    \pgfmathsetmacro{\xoney}{\umax + \d}
    \pgfmathsetmacro{\xtwox}{\umax + \d}
    \pgfmathsetmacro{\xtwoy}{\xmax - \a*(\umax + \d)}

    \pgfmathsetmacro{\yonex}{\uone - \d}
    \pgfmathsetmacro{\yoney}{\umax + \d}
    \pgfmathsetmacro{\ytwox}{\umax + \d}
    \pgfmathsetmacro{\ytwoy}{\uone - \d}

    \pgfmathsetmacro{\zonex}{\yonex + (\utwo - \uone)/2}
    \pgfmathsetmacro{\zoney}{\yoney + (\utwo - \uone)/2}
    \pgfmathsetmacro{\ztwox}{\ytwox + (\utwo - \uone)/2}
    \pgfmathsetmacro{\ztwoy}{\ytwoy + (\utwo - \uone)/2}

    \pgfmathsetmacro{\midx}{\umax / 3}
    \pgfmathsetmacro{\midy}{\umax / 3}

    \coordinate (b1) at (0,\umax);
    \coordinate (b2) at (\uone,\umax);
    \coordinate (b3) at (\umax,\utwo);
    \coordinate (b4) at (\umax,0);
    \coordinate (b5) at (0,0);

    \coordinate (x1) at (\xonex, \xoney);
    \coordinate (x2) at (\xtwox, \xtwoy);
    \coordinate (y1) at (\yonex, \yoney);
    \coordinate (y2) at (\ytwox, \ytwoy);
    \coordinate (z1) at (\zonex, \zoney);
    \coordinate (z2) at (\ztwox, \ztwoy);

    \coordinate (vu) at (\umax, \utwo);
    \coordinate (vl) at (\uone, \umax);

    \coordinate (mid) at (\midx, \midy);

  \draw[thick, green!80, opacity=0.3] (b1) -- (b2) -- (b3) -- (b4);

  \fill[green!80, opacity=0.3] (b1) -- (b2) -- (b3) -- (b4) -- (b5) -- cycle;

\draw[ultra thick, red!100!black] (b2) -- (b3);
\draw[blue!80!black, dashed, thick] (y1) -- (y2);
\node[below right, text=red!90!black, xshift=4.5em] at (b2) {$\overline{\mathcal{H}}$};
\node[above right, text=blue!70!black, xshift=-1.5em, yshift=2em] at (y2) {$\overline{y}^*(t) = \sum_{s=1}^tu(s)$};
\node[teal!60!black, right] at (mid) {$\mathcal{F}^{t}_i$};




\end{tikzpicture}
    \caption{$\overline{y}^*(t)$ is found by minimizing $\sum_s^tu(t)$ s.t. $u \in \overline{\mathcal{H}} \cap \mathcal{F}^{t}_{ i}$}
    \label{fig:visual_lemma}
\end{figure}
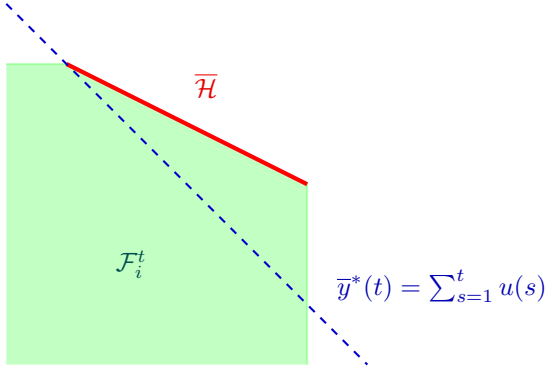

\subsection{Inner Approximation}
We now consider the construction of a maximal inner approximation of the feasible set $\mathcal{F}_{ i}$ using the parameterized base polytope $\mathcal{B}_i(\underline{y}, \overline{y})$. To proceed we will show how one can iteratively compute the values of $\underline{y}$ $\overline{y}$, in the following we will denote their optimal values as $\underline{y}^*$ and $\overline{y}^*$.
We begin by introducing the following notation: let $\mathcal{B}^t(\underline{y}, \overline{y})$ denote the restriction of the base polytope \(\mathcal{B}(\underline{y}, \overline{y})\) to its first \(t\) components, i.e.
\begin{equation}
    \mathcal{B}^t(\underline{y}, \overline{y}) := \left\{ u_{[1:t]} \in \mathbb{R}^t \,\middle|\, u \in \mathcal{B}(\underline{y}, \overline{y}) \right\},
\end{equation}
We will use similar notation for $\mathcal{F}_{ i}$. 
\begin{lem}\label{lemma:battery}
    For all $t \in \mathcal{T}$, the optimal parameters of $\mathcal{B}(\underline{y}, \overline{y})$ are the solutions to
    \begin{equation}\label{eq:opt_under}
    \begin{aligned} 
    \underline{y}^*(t) = \min_{\underline{y}(t)} \quad & \underline{y}(t) \\ 
    \text{s.t.} \quad & \underline{y}(t) \leq \sum_{s=1}^t u(s)\\
     &\quad \forall u \in \mathcal{F}^{t}_{ i} 
     \end{aligned}
     \end{equation}
\begin{equation} \label{eq:opt_over}
    \begin{aligned} 
    \overline{y}^*(t) = \max_{\overline{y}(t)} \quad & \overline{y}(t) \\ 
    \text{s.t.} \quad & \sum_{s=1}^t u(s) \leq  \overline{y}(t) \\
     &\quad \forall u \in \mathcal{F}^{t}_{ i} \end{aligned} \end{equation}

\end{lem}
The proof can be found in the Appendix. We can reformulate \eqref{eq:opt_under} and \eqref{eq:opt_over} with the following Lemma.

\begin{lem}\label{lemma:reformulation}
The computation of the optimal lower bound \( \underline{y}^*(t) \), as defined in \eqref{eq:opt_under}, can be equivalently reformulated as the following optimization problem:
\begin{equation}\label{eq:reform_under}
\begin{aligned}
    \underline{y}^*(t) = \max_{u \in \mathbb{R}^t} \quad & \sum_{s=1}^t u(s) \\
    \text{s.t.} \quad & u \in \underline{\mathcal{H}}, \\
    & u \in \mathcal{F}^t_{ i},
\end{aligned}
\end{equation}
where the hyperplane \( \underline{\mathcal{H}} \subset \mathbb{R}^t \) is defined as
\[
    \underline{\mathcal{H}} := \left\{ u \in \mathbb{R}^t \;\middle|\; \sum_{s=1}^t a^{t-s} u(s) = \underline{z}_i(t) \right\},
\]
and \( \underline{z}_i(t) \) is given by
\[
    \underline{z}_i(t) := \max \left\{ \underline{x}_i(t), \sum_{s=1}^t a^{t-s} \underline{u}_i \right\}.
\]

Similarly, the optimal upper bound \( \overline{y}^*(t) \), corresponding to \eqref{eq:opt_over}, can be reformulated as:
\begin{equation}\label{eq:reform_over}
\begin{aligned}
    \overline{y}^*(t) = \min_{u \in \mathbb{R}^t} \quad & \sum_{s=1}^t u(s) \\
    \text{s.t.} \quad & u \in \overline{\mathcal{H}}, \\
    & u \in \mathcal{F}^t_{ i},
\end{aligned}
\end{equation}
where the hyperplane \( \overline{\mathcal{H}} \subset \mathbb{R}^t \) is defined as
\[
    \overline{\mathcal{H}} := \left\{ u \in \mathbb{R}^t \;\middle|\; \sum_{s=1}^t a^{t-s} u(s) = \overline{z}_i(t) \right\},
\]
with
\[
    \overline{z}_i(t) := \min \left\{ \overline{x}_i(t), \sum_{s=1}^t a^{t-s} \overline{u}_i \right\}.
\]
\end{lem}
The proof of Lemma~\ref{lemma:reformulation} is provided in the Appendix. Fig. \ref{fig:visual_lemma} provides intuition behind this Lemma, when considering the upper bound $\overline{y}^*(t)$.
Concluding this section, we consolidate the preceding results and present a construction of the maximal inner approximation with the following theorem.
\begin{thm}
    The maximal inner approximation of the set $\mathcal{F}_{ i}$ using a base polytope representation is given by $\mathcal{B}(\underline{y}^*, \overline{y}^*)$, where the optimal lower and upper bounds, $\underline{y}^*$ and $\overline{y}^*$, are computed using \eqref{eq:reform_under} and \eqref{eq:reform_over} respectively.
\end{thm}
\begin{proof}
    This follows trivially from using Lemma \ref{lemma:reformulation} with Lemma \ref{lemma:battery}.
\end{proof}


\subsection{Aggregation via g-polymatroids}
We now show how the inner approximation  $\mathcal{B}(\underline{y}_i^*, \overline{y}_i^*)$ can be characterized as a g-polymatroid, and derive the super- and submodular functions that generate the g-polymatroids. Given this characterization, we conclude by giving an inner approximation of the aggregate flexibility for a collection of TCLs.

\begin{lem}\label{lemma:g-poly}
    $\mathcal{B}(\underline{y}^*, \overline{y}^*)$ is the g-polymatroid $\mathcal{Q}(p_i^T, b_i^T)$ where the super- and submodular functions can be computed using the following recursion for all $t \in \mathcal{T}$:
            \begin{equation*}
            \begin{aligned}
                p_i^t(A) =    \;& \underline{u}(A \setminus [t]) \\ 
                &+ \max \left\{ 
                    \begin{array}{c} 
                        p^{t-1}(A \cap [t]), \\ 
                        \underline{y}^*(t) - b^{t-1}(A') + \overline{u}(A' \setminus [t])
                    \end{array}
                    \right\} 
            \end{aligned}
        \end{equation*}
        \begin{equation*}
            \begin{aligned}
            b_i^t(A) =    \;& \overline{u}(A \setminus [t]) \\
                        & + \min \left\{ 
                                    \begin{array}{c} 
                                        b^{t-1}(A \cap [t]), \\ 
                                         \overline{y}^*(t) - p^{t-1}(A') + \underline{u}(A' \setminus [t])
                                    \end{array}
                            \right\}
            \end{aligned}
        \end{equation*}
\end{lem}
\begin{proof}
    Lemma~1 of \cite{Mukhi2025ExactPolymatroids} provides a characterization of polytopes of the form  \( \mathcal{B}(\underline{y}^*, \overline{y}^*) \) as a g-polymatroids, and Corollary~1 of \cite{Mukhi2025ExactPolymatroids} gives the corresponding super- and submodular functions. 
    The result follows trivially by substituting the parameters of the optimal base polytope \( \mathcal{B}(\underline{y}^*, \overline{y}^*) \) into these results.
\end{proof}
Finally, leveraging the characterization of the inner approximations as g-polymatroids, together with Theorem~\ref{thm:g-poly_m_sum}, we obtain, $\mathcal{Q}(p^*_\mathcal{N}, b^*_\mathcal{N})$, a representation of an inner approximation to the aggregate flexibility set of a collection of TCLs:
\begin{equation*}
        \mathcal{Q}(p^*_\mathcal{N}, b^*_\mathcal{N}) \subseteq \mathcal{F}_\mathcal{N}
\end{equation*}
where the super- and submodular functions of $\mathcal{Q}(p_\mathcal{N}, b_\mathcal{N})$, are given by
\begin{equation*}
        p_\mathcal{N} := \sum_i^N p_i^T,  \quad  b_\mathcal{N} := \sum_i^N b_i^T,
\end{equation*}
and $p_i^T$ and $b_i^T$ are given by Lemma \ref{lemma:g-poly}.

\section{Numerical Results}\label{sec:num_res}
This section presents numerical results that benchmark the performance of the proposed aggregation methods against existing approaches in the literature. We begin by quantifying the approximation error associated with the proposed techniques. We conclude with a case study to highlight the real-world applicability of the proposed approach. For brevity, only these experiments are presented; however, the proposed approximation methods are also favorable in terms of computational complexity.
To benchmark, we compare our results with the homothet-based approximations introduced in \cite{Zhao2017ALoads} and the zonotope-based approaches developed in \cite{Muller2019AggregationResources}.

\subsection{Approximation Error}
\begin{figure}
    \centering
    \includegraphics[width=0.98\linewidth]{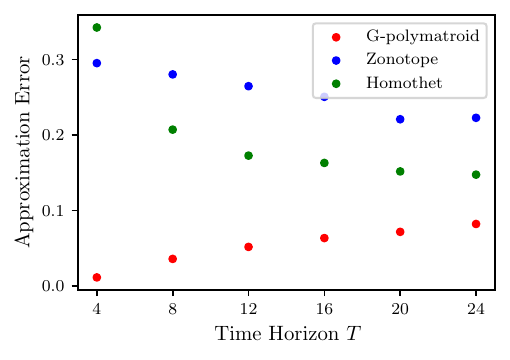}
    \caption{Approximation quality of the different methods as a function of the time horizon length. The proposed method consistently achieves lower approximation error compared to benchmark approaches.}
    \label{fig:approx_error}
\end{figure}

In this part, we compare the quality of the approximation methods. To this end, we sample a population of 100 TCLs and construct the corresponding approximations. A cost vector \( c \in \mathbb{R}^{\mathcal{T}} \) is generated, where each element is uniformly distribution over $[0,1]$. For each approximation, we formulate and solve a linear program (LP) to minimize the cost over the respective feasible set. The approximation error is then computed as
\begin{equation}
    \text{Approximation Error} = \frac{J^*_{\text{approx}} - J^*_{\text{exact}}}{J^*_{\text{exact}}},
\end{equation}
where \( J^*_{\text{approx}} \) denotes the optimal cost obtained using the approximation. \( J^*_{\text{exact}} \) denotes the true optimal cost, computed by individually optimizing each device and summing their respective costs.
As solving a linear programs selects vertices of the feasible set, the difference between the approximate and exact optimal costs can be used as a proxy for the approximation quality. For a perfect approximation, the approximation error vanishes, since the approximate and exact feasible sets would be identical. We repeat this process for multiple realizations of the cost vector and TCL populations, across varying time horizon lengths. The results are shown in Fig. \ref{fig:approx_error}, where we see the approximation methods proposed here, out perform those in the literature. Note however, the quality of the approximation decreases with the length of the time horizon.

\subsection{Tracking generation signals}
\begin{figure}
    \centering
    \includegraphics[width=0.98\linewidth]{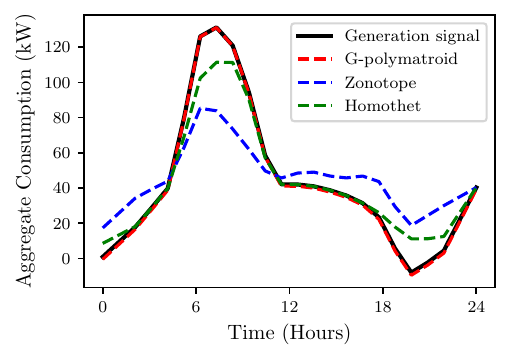}
    \caption{Aggregate charging profile during generation signal tracking. The proposed method follows the signal more accurately than the two benchmark approaches.}
    \label{fig:tracking_gen_sig}
\end{figure}
To demonstrate the practical utility of the proposed approach, we consider the task of tracking a reference generation signal. Specifically, we simulate a population of 100 TCLs, and synthesized a generation signal, $g \in \mathbb{R}^\mathcal{T}$. The objective is to ensure that the aggregated load profile tracks this signal. To achieve this, we formulate an optimization problem in which the objective is to minimize the $\ell^2$-norm between the aggregate consumption profile and $g$.
Fig. \ref{fig:tracking_gen_sig}  illustrates the optimal aggregate consumption profiles obtained using each of the considered aggregation methods. The approximation method presented in this paper forms a tighter approximation than the other methods in the literature, and thus is able to track the generation signal more closely.

\section{Conclusions and Future Work}\label{sec:conc}
In this paper, we introduced a novel method for approximating the flexibility sets of individual TCLs using g-polymatroids. Specifically, we employed a base polytope that is a g-polymatroid, and formulated an optimization problem to compute its maximal inner approximation within the true flexibility set of a TCL. From this, and using the properties of g-polymatroids we were able to provide representations of the aggregate flexibility in large populations of TCLs. Numerical results showed this approach offers tight approximations compared with other methods in the literature.

This work is particularly significant in the context of prior work that uses g-polymatroids to exactly characterize the flexibility sets of a less general class of devices, i.e. those with lossless storage dynamics. In conjunction with these earlier results, the proposed approach contributes to a broader framework for accurately and scalably characterizing the flexibility of a wide range of energy-constrained devices.

Future work may explore the use of alternative g-polymatroids as the base polytopes for constructing inner approximations. There exists a broad class of polytopes that exhibit g-polymatroid structure, some of which may be even better suited to capturing the dissipative dynamics of TCLs.

\bibliographystyle{IEEEtran}
\bibliography{references}
\appendix
\section{Proofs}
\subsection{Lemma \ref{lemma:battery}}
\begin{proof}
    For all $t \in \mathcal{T}$, if $\mathcal{B}^{t-1}(\underline{y}, \overline{y}) \subseteq \mathcal{F}^{t-1}_{ i}$ and 
    \begin{equation}\label{eq:lmma_inc}
        \underline{y}(t) \leq \sum_{s=1}^t u(s) \leq \overline{y}(t) \quad \forall \; u \in \mathcal{F}^{t}_{ i},
    \end{equation}
    then $\mathcal{B}^{t}(\underline{y}, \overline{y}) \subseteq \mathcal{F}^{t}_{ i}$, ensuring  $\mathcal{B}^{t}(\underline{y}, \overline{y})$ is an inner approximation. Furthermore, if $\underline{y}_b(t) \leq \underline{y}_a(t)$ and $\overline{y}_a(t) \leq \overline{y}_b(t)  \;\;\forall t \in \mathcal{T}$  then $\mathcal{B}(\underline{y}_a, \overline{y}_a) \subseteq \mathcal{B}(\underline{y}_b, \overline{y}_b)$. Therefore finding the maximum for $\overline{y}(t)$ and minimum for $\underline{y}(t)$, provides us with the maximal inner approximation. 
\end{proof}
\subsection{Lemma \ref{lemma:reformulation}}
\begin{proof}
    When maximizing $u \in \mathcal{F}_{ i}$ either one of the following two constraints will be active at the optimum:
    \begin{align*}
         \sum_{s=1}^t a^{t-s}u(s) &\leq \overline{x}_i(t) \\
         u(t) &\leq \overline{u}_i  \quad \forall \; t\in \mathcal{T}.
    \end{align*}
    Combing these two, we define a plane,  $\mathcal{H}$, where either of these constraints is tight:
    \begin{equation*}
        \overline{\mathcal{H}}:= \left\{ u\in \mathbb{R}^t \; \mid \sum_{s=1}^t a^{t-s}u(s) = \overline{z}_i(t)\right\},
    \end{equation*}
    where $\overline{z}_i(t) = \textrm{min} \left\{\overline{x}_i(t), \sum_{s=1}^t a^{t-s}\overline{u}_i \right\}$. 
    Now consider the problem:
    \begin{align*}
    u^* = \underset{u}{\textrm{argmin}} \quad & \sum_{s=1}^t u(s) \\
    \text{s.t.} \quad  & u\in \overline{\mathcal{H}}\\
    & u \in \mathcal{F}_{ i}. \\
\end{align*}
    By construction we have $\sum_{s=1}^t u^*(s) \leq  \sum_{s=1}^t u(s) \quad \forall u \in \mathcal{F}_{ i}$. Setting $\overline{y}^*(t) =  \sum_{s=1}^t u(s)$ we get the desired result for the upper bound on $\overline{y}(t)$.

    Similarly when to obtain the lower bound  $\underline{y}(t)$ we minimize  $u \in \mathcal{F}_{ i}$, where the active constraints are:
        \begin{align*}
         \overline{x}_i(t)  \leq \sum_{s=1}^t a^{t-s}u(s) \\
          \underline{u}_i \leq u(s) \quad \forall \; t\in \mathcal{T}.
    \end{align*}
    The plane defining the active constraints is then
        \begin{equation*}
        \underline{\mathcal{H}}:= \left\{ u\in \mathbb{R}^t \; \mid \sum_{s=1}^t a^{t-s}u(s) = \underline{z}_i(t)\right\},
    \end{equation*}
    where $\underline{z}_i(t) = \textrm{max} \left\{\underline{x}_i(t), \sum_{s=1}^t a^{t-s}\underline{u}_i \right\}$. As before for the optimal solution of the problem
    \begin{align*}
    u^* = \underset{u}{\textrm{argmax}} \quad & \sum_{s=1}^t u(s) \\
    \text{s.t.} \quad  & u\in \underline{\mathcal{H}}\\
    & u \in \mathcal{F}_{ i}. \\
\end{align*}
    we have $ \sum_{s=1}^t u(s) \leq \sum_{s=1}^t u^*(s)  \quad \forall u \in \mathcal{F}_{ i}$. Setting $\underline{y}(t) =  \sum_{s=1}^t u^*(s)$ we conclude our proof.
\end{proof}

\end{document}